\documentclass{article}

\usepackage{microtype}
\usepackage{graphicx}
\usepackage{subcaption}
\usepackage{booktabs} 

\usepackage{hyperref}

\usepackage{xcolor}
\newcommand{\hlc}[2]{\colorbox{#1}{#2}}
\definecolor{lgreen}{HTML}{B2CEB5}
\definecolor{lred}{HTML}{EBB6B5}
\definecolor{lpurple}{HTML}{C0B4DC}

\usepackage{tikz}
\usetikzlibrary{matrix,decorations.pathreplacing} 

\usepackage{algorithm}

\newcommand{\sinv}{\sigma^{-1}}

\usepackage[preprint]{icml2026}

\usepackage{amsmath}
\usepackage{amssymb}
\usepackage{mathtools}
\usepackage{amsthm}
\usepackage{thm-restate}

\DeclareMathOperator*{\argmax}{arg\,max}

\usepackage[capitalize,noabbrev]{cleveref}

\theoremstyle{plain}
\newtheorem{theorem}{Theorem}[section]

\newtheorem{corollary}[theorem]{Corollary}
\theoremstyle{definition}
\newtheorem{definition}[theorem]{Definition}

\theoremstyle{remark}

\newtheorem{observation}[theorem]{Observation}
\newtheorem{empobservation}[theorem]{Empirical Observation}

\usepackage[textsize=tiny]{todonotes}

\icmltitlerunning{Information Loss and Disparate Effects in Network Embeddings}

\begin{document}

\twocolumn[
  \icmltitle{Information Loss and Disparate Effects in Network Embeddings}
  \icmlsetsymbol{equal}{*}

  \begin{icmlauthorlist}
    \icmlauthor{Gabriel Chuang}{col}
    \icmlauthor{Augustin Chaintreau}{col}
  \end{icmlauthorlist}

  \icmlaffiliation{col}{Department of Computer Science, Columbia University, New York, USA}

  \icmlcorrespondingauthor{Gabriel Chuang}{gtc2117@columbia.edu}

  \icmlkeywords{Machine Learning, Fairness, Network Embeddings}

  \vskip 0.3in
]



\printAffiliationsAndNotice{}  

\begin{abstract}
An extensive line of work studies fairness interventions for network embeddings, but less is known about their baseline behavior. In this work, we ask: how do baseline embeddings (without fairness interventions) produce disparate effects at the representation level? We analyze the asymptotic behavior of low-dimensional embeddings on stochastic block model (SBM) graphs, which encode both homophily and group structure. We characterize exact conditions under which embeddings cause information loss, showing that the amount of information loss depends directly on the graph’s density and assortativity. Notably, very different graphs can produce identical embeddings in the limit, and this non-invertibility disproportionately affects smaller and sparser communities. As a result, simple downstream tasks, such as link prediction, introduce higher error rates for these communities, helping explain disparities widely observed in practice.
\end{abstract}

\section{Introduction}
Low-dimensional network embeddings are powerful tools for doing machine learning tasks on graph or network data
(e.g., \cite{perozzi2014deepwalk,tang2015line,grover2016node2vec,qiu2018network}, for an overview, see \cite{cui2018survey}). 
These techniques have shown extremely strong performance in tasks such as link prediction, community detection, and node classification. 
A line of work aims to encode notions of fairness into network embeddings, including FairWalk \cite{rahman2019fairwalk}, Debayes \cite{buyl2020debayes}, and many others \cite{khajehnejad2022crosswalk, spinelli2021fairdrop}. 
These techniques achieve empirical success at various measures of fairness, such as group-level statistical parity. This series of work is predicated on the empirical observation that baseline network embeddings often have unequal error rates across groups.

Orthogonally, recent work has shown theoretical limits on the representation power of arbitrary low-dimensional embeddings \cite{seshadhri2023limitations, seshadhri2020impossibility, snoeck2025compressibility}: any embedding must lose some information from the underlying graph. Davison \& Austern \yrcite{davison2023asymptotics} show that the the training algorithm is another factor in information loss, showing examples of very low-rank graphons that cannot be learned even by high-dimensional embeddings if trained using subsampling and dot-product based loss functions. That is, the ``default'' way of training network embeddings, by using inner products, maps multiple distinct graphs (or graphons) to identical embeddings: these embeddings are non-invertible.

In this work, we connect these two threads, showing that this information-loss-induced non-invertibility is closely related to the disparate error rates observed in the fairness literature. 

Like Davison \& Austern \yrcite{davison2023asymptotics}, we consider learning the generative model rather than a specific instance (i.e., edge set) generated from the model. In particular, we focus on stochastic block model (SBM) graphs \cite{holland1983stochastic}, which capture homophily (nodes tend to be connected to similar nodes) and majority/minority group structure, while being simple enough to make analysis tractable. Similarly, we focus on relatively simple dot-product-based embedding schemes, which allow for cleaner theoretical analysis.

\begin{figure*}[!tbp]
    \centering
    \includegraphics[width=0.9\linewidth]{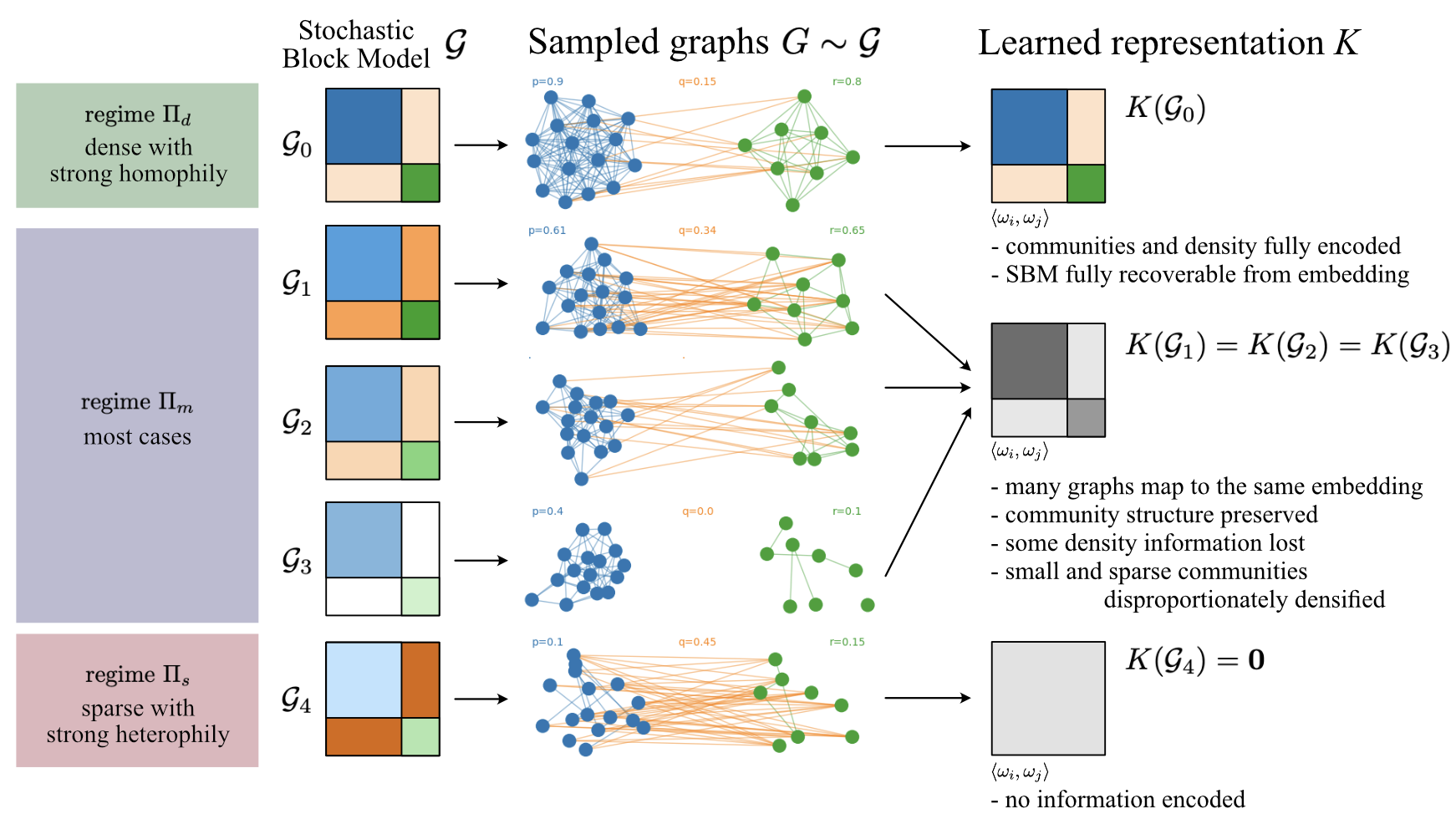}
    \caption{We characterize conditions under which the embedding process encodes full, partial, or no information about the input stochastic block model (SBM). In a \hlc{lgreen}{dense regime $\Pi_d$}, the embedding learns (perfectly) the information of the SBM. In a \hlc{lred}{sparse regime $\Pi_s$}, the embedding is degenerate and loses all information. In \hlc{lpurple}{most cases $\Pi_m$}, community structure is encoded but only some edge density information is retained; in this regime, very different SBMs map to the same representation in a way that leads to disparate group-level effects.}
    \label{fig:bigfig}
\end{figure*}

\subsection{Contributions.}

Our main contributions are as follows. 

\begin{enumerate}
    \item The amount of information lost when embedding a SBM depends on its density and assortativity, but \emph{not} the sizes of the communities. \textbf{Dense, assortative graphs preserve all information; heterophilic and sparse SBMs lose all information; and most graphs lie in an intermediate regime where only some information is preserved.} We exactly characterize these regimes in the 2-community case (Theorem~\ref{thm:embedding-limit-regimes}).
    \item In the intermediate regime, many SBMs receive identical embeddings, inducing a set of equivalence classes. We show that each class contains \emph{very} different graphs, ranging from very sparse to very dense; as such, \textbf{sparse graphs can ``look like'' much denser graphs when embedded (and vice versa)}. We give an analytic characterization of these classes for 2-community SBMs (Theorem~\ref{thm:eqclasses}).
    \item \textbf{Communities are disparately affected by non-invertibility, based on size and sparsity}.  Smaller and sparser communities are disproportionately ``densified'' by generally embedding them to smaller, more tightly clustered regions of representation space. This introduces greater error for these communities in tasks such as link prediction (by overpredicting edge probabilities of these smaller/sparser groups). This particular mechanism of disparate error is visible in past work but is not (to our knowledge) explicitly pointed out. This leads directly to disparate error rates in downstream tasks like link prediction. 
\end{enumerate}




\section{Background} 

\subsection{Limitations of low-dimensional graph representations}

For a survey of recent work on the limitations of graph embeddings, see \cite{seshadhri2023limitations}. For example, \cite{seshadhri2020impossibility} show that low-dimensional inner product based embeddings cannot encode both low degree and high clustering coefficient, and \cite{snoeck2025compressibility} shows that neighborhood preservation in arbitrary metric spaces requires very high embedding dimension. \cite{chanpuriya2021deepwalking} develop a method for approximately \emph{inverting} graph embeddings (that is, finding a graph that maps to a given embedding); they find that both global properties (like triangle density) and local properties (like existence of a particular edge) are often lost, but that the community structure is typically preserved and sometimes even enhanced (i.e., the conductance within the community is decreased). Our work provides a potential explanation for this behavior, by characterizing a family of SBM graphs (of varying community density) that are embedded identically. On the other hand, Chanpuriya et al \yrcite{chanpuriya2020node} show that many of these representation limits are a consequence of learning only \emph{one} embedding vector for each node, and prove that learning ``word'' and ``context'' vectors allow for substantially more expressive embeddings. 


A recent paper by Davison \& Austern \yrcite{davison2023asymptotics} develops methods for showing the limiting distributions of embeddings that are based on minimizing subsampling-based losses on exchangeable graphs (i.e., graphons). They show that unique loss-minimizing embeddings exist (up to symmetry), and prove statistical rates of convergence to these asymptotic minimizers. 
They show a simple case where an inner-product-based embedding is non-invertible, and prove that using an alternative indefinite inner product resolves these invertibility problems in a way that the standard inner product cannot. However, the standard inner product is likely to remain prevalent in practice, so their work leaves unanswered many important questions about the behavior of representations in practice, which we explore in this work. 




\subsection{Fair network embeddings}

It is well-established that network algorithms often introduce or reinforce structural fairness issues (e.g., \cite{masrour2020bursting}; for a survey, see \cite{dong2023fairness}), and a long line of work aims to address this in network embeddings, including by modifying the sampling procedure \cite{rahman2019fairwalk}, preprocessing the graph \cite{spinelli2021fairdrop} adding constraints to the objective function \cite{bose2019compositional}, and many others (e.g., \cite{buyl2020debayes, li2021dyadic}. Many of these works propose measures of fairness, either for the embeddings themselves or via downstream tasks. In this work, we do \emph{not} consider a particular measure of fairness, nor do we consider any particular fairness intervention. Instead, we seek to understand the mechanism by which baseline algorithms introduce disparities.  

In parallel, there has been substantial work on theoretical expressivity limits on Graph Neural Networks (GNNs), for example, see \cite{garg2020generalization}. The fundamentally different task-specific objectives and structure of GNNs lead to expressivity results that are orthogonal to what we describe here. 
\subsection{Stochastic Block Models} 
Stochastic block models (SBMs) \cite{holland1983stochastic} are a popular generative model for graphs, generalizing standard Erd\H{o}s-Renyi graphs to represent community structure (each ``block'' is a community) and homophily (a homophilous SBM has higher within-block probabilities than between-block probabilities). 

\begin{definition}[$k$-block Stochastic Block Model]
    A $k$-block stochastic block model is a generative model for graphs $\mathcal G = (P, S)$, where $P$ is a $k\times k$ symmetric matrix of block edge probabilities, and $S$ is a list of block sizes $S = [s_1, \cdots s_k]$ (adding up to 1). 
\end{definition}

\subsection{Our Setting}
Throughout, we will use $\sigma(x) = \frac{e^x}{1+e^x}$ to denote the standard logistic function and $\sigma^{-1}(x) = \ln(x)-\ln(1-x)$ to denote its inverse (the logit function). We use $\langle \cdot, \cdot \rangle$ to denote the standard inner product. 

Following \cite{davison2023asymptotics}, we consider the general node embedding algorithm shown in Algorithm~\ref{alg:embedding_generic}, which depends on a choice of stochastic subsampling scheme $S$ and edge-wise loss function $\ell$ that depends on the inner products of embeddings. Many network embedding algorithms (e.g., node2vec \cite{grover2016node2vec}, DeepWalk \cite{perozzi2014deepwalk}) are instantiations of this structure, mainly with variations in the sampling scheme $S$. 

\begin{algorithm}[!h]
\caption{Node embedding for an SBM (generic)}
\textbf{Params:} Graph subsampling scheme $S$, loss function $\ell$ \\
\textbf{Input:}  SBM $\mathcal G$, embedding dimension $d$, graph size $n$ \\
\textbf{Output:} vertex embeddings $\omega_i \in \mathbb R^d : i = 1, \cdots |V|$ \\
\begin{algorithmic}[1]
\STATE Initialize $w_i \in \mathbb R^d$ in some way (e.g., uniform) 
\WHILE{convergence not reached}
    \STATE Sample graph $G \sim_n \mathcal G$ with $G = (V,E), |V| = n$.
    \STATE Sample subgraph $s = S(G)$
    \STATE Compute loss $\ell(s, (\omega_1, \cdots \omega_{|V|}))$ 
    \STATE Perform a gradient update to $\{\omega_1, \cdots, \omega_n\}$.
\ENDWHILE 
\end{algorithmic}
\label{alg:embedding_generic}
\end{algorithm}

We use the following choice of sampling scheme and loss function: 

\begin{definition}[Uniform Vertex Sampling]
The graph subsampling scheme $S$ uniformly selects $v$ vertices of $G$ and returns the induced subgraph. 
\end{definition}
\begin{definition}[Cross-Entropy Loss]
    $\ell$ sums the pairwise cross-entropy loss of each pair of vertices. $\ell_\text{pw}(y,x) = -x\log(\sigma(y)) - (1-x) \log (1-\sigma(y))$.
\end{definition}


Algorithm~\ref{alg:embedding_generic} minimizes the empirical risk function 
\begin{equation}
  \begin{aligned}
    &\mathcal{R}_n(\omega_1, \cdots, \omega_n) =  \\ 
    &\sum_{i,j \in [n], i\neq j}  \bigg(\Pr((i,j) \in S(G_n)\mid G_n) \ell_{\text{pw}}(\langle\omega_i, \omega_j\rangle, a_{ij})\bigg).
    \end{aligned}
    \label{eq:empirical-risk}
\end{equation}

Davison \& Austern \yrcite{davison2023asymptotics} show that the minimizers of the empirical risk (\ref{eq:empirical-risk}) are unique up to symmetry/rotation, i.e., that, for a given graph $\mathcal G_n$ and sampling scheme $S$, there exists a unique inner-product (kernel) matrix $K(\mathcal G) \in \mathbb R^{n\times n}$ such that all loss-minimizing $\hat \omega_i$ satisfy 
\[\frac{1}{n^2} \sum_{i,j} |\langle \hat \omega_i, \hat \omega_j \rangle - K(\mathcal G)_{i,j}| = o_p(1). \]
That is, it converges to zero with high probability. 

In this sense, $K(\mathcal G)$ is ``the embedding of $\mathcal G$'' just as much as $\{\omega_i\}_{i \in [n]}$ is: since training depends only on the pairwise inner products, the unique minimizer is fully captured by the kernel. Therefore, we seek to analytically characterize the kernel $K(\mathcal G)$ of these limiting embeddings, find closed forms (where they exist) and determine what families of graphs have the same $K$. This allows us to be agnostic to the dimension of the embedding: our results, like those of \cite{davison2023asymptotics}, are true even for arbitrarily high dimension.

\section{A motivating example}

\begin{figure}
    \centering
    \includegraphics[width=0.49\linewidth]{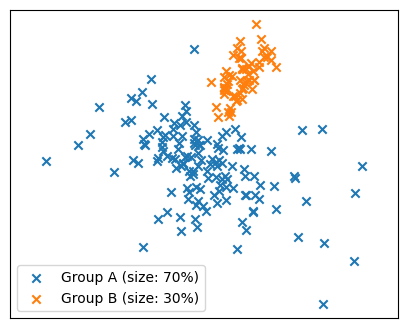}
    \includegraphics[width=0.49\linewidth]{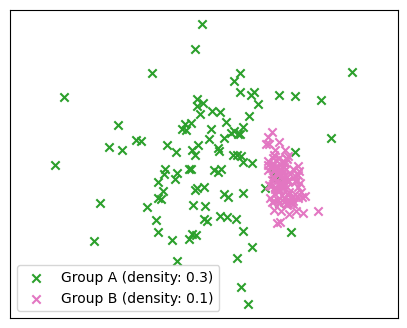}
    \caption{Two motivating examples. \emph{(a), left:} embedding two same-density, different-size groups gives a representation that makes the smaller group seem much more dense. \emph{(b), right:} similarly, sparser groups are made to appear more dense.}
    \label{fig:motivating}
\end{figure}

Suppose we want to embed the two-block SBM $p = [0.1, 0; 0, 0.1] 
, S = (0.7, 0.3)$: two disconnected, equally-dense communities of different sizes. Intuitively, we could expect something like Fig.~\ref{fig:motivating}(a): the larger community gets more ``volume'' in embedding space. This is, in fact, the minimizer of the loss (Eq.~\ref{eq:empirical-risk}). The difference in group sizes causes the smaller community to be ``squished'' into a smaller region of space, \emph{even though both communities are equally dense in the original graph}. 

We can intuitively see why this is the case by imagining a counterfactual where both communities receive equal amounts of ``volume''. Since the graph is relatively sparse (95\% non-edges), most of the loss comes from non-edges whose endpoints are nevertheless close in embedding space. The larger group has many more such non-edges (scaling quadratically), so the marginal benefit of giving the larger group more ``volume'' is high compared to the marginal cost of squishing the smaller group. 

Counterintuitively, the same ``squishing'' effect also occurs when one group is \emph{sparser} than the other (shown in Fig.~\ref{fig:motivating}(b)): the sparser community gets less space and so \emph{also} appears denser. 

Although it is not pointed out explicitly, evidence of this can be seen by revisiting the baseline results throughout the fair embeddings literature. For example,  in FairWalk \cite{rahman2019fairwalk}, the ``baseline'' results show that the edges predicted for the minority group (gender 1, 37\% of the population) are twice as dense as for the majority group, despite having comparable initial densities. 

This intuition is what we seek to understand in a precise way in the following sections. When do particular communities lose information when embedded, and at what rates? What governs this information loss? 



\section{Information Loss}

\subsection{Exact characterizations for 2-community SBMs}
First, we consider two-block SBMs. For simplicity, we will write the SBM as $P = \begin{bmatrix}
    p & q \\ q & r
\end{bmatrix}, S = (a, 1-a)$ and let $a \geq \frac12$: that is, the graph has one larger group (size $a$) with intra-group probability $p$; one smaller group (size $1-a$) with probability $r$; and inter-group probability $q$. 

There are three regimes of information loss as a function of $(p,q,r)$, depicted in Fig.~\ref{fig:regions-split}: 

\begin{figure}
    \begin{subfigure}{0.25\linewidth}
        \includegraphics[width=\linewidth]{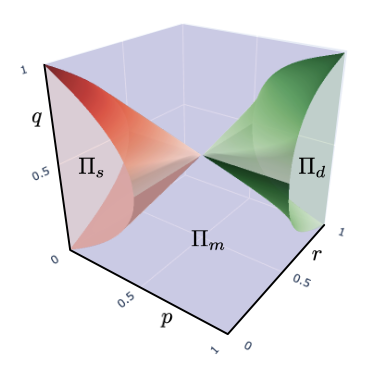}
    \end{subfigure}
    ~
    \begin{subfigure}{0.75\linewidth}
        \includegraphics[width=\linewidth]{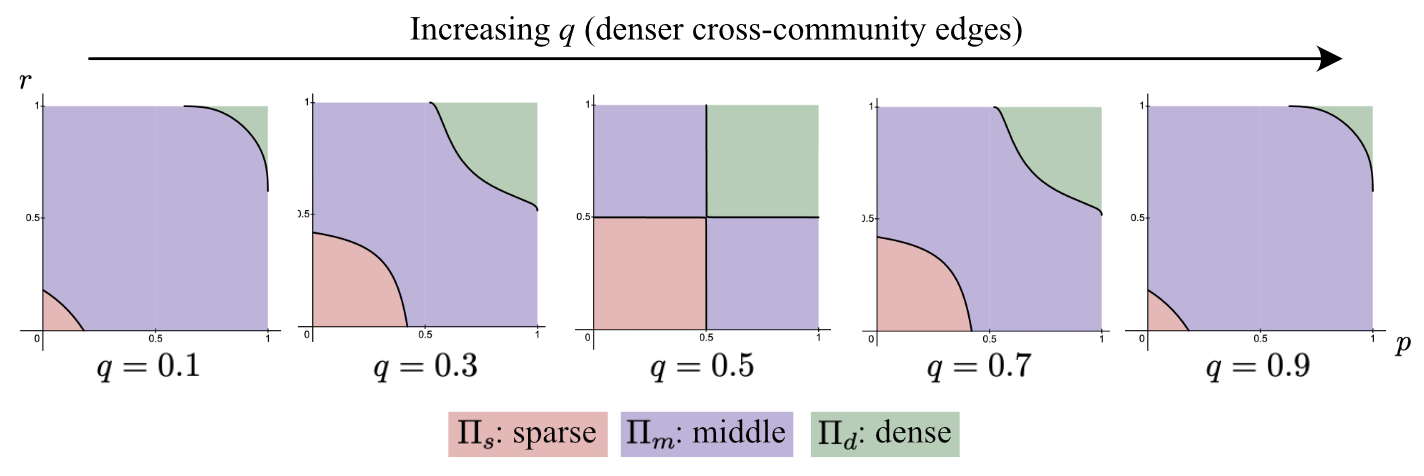}
    \end{subfigure}

    \centering
    \caption{Regions of 2-block SBM parameter space in which \hlc{lred}{none}, \hlc{lpurple}{some}, or \hlc{lgreen}{all} of the information is preserved in the embedding.}
    \label{fig:regions-split}
\end{figure}

\begin{itemize}
    \item $\Pi_d$: $p \geq \frac12, r \geq \frac12$, and  $\sigma^{-1}(q)^2 \leq \sigma^{-1}(p) \sigma^{-1}(r)$  (the ``dense, assortative'' regime);
    \item $\Pi_s$: $p \leq \frac12, r\leq \frac12$, and $(\frac12-q)^2 \leq (\frac12-p)(\frac12-r)$ (the ``sparse, heterophilic'' regime); and 
    \item $\Pi_m$: $[0,1]^3 \setminus \Pi_d \setminus \Pi_s$, (the ``middle'' regime).
\end{itemize}

Our results can be summarized as follows (with formal statement in Theorem 1): 

\textbf{Dense communities $\Pi_d$ (Theorem 1.1):} When the two communities are sufficiently dense (many intra-community edges relative to the between-community edges) and the inter-community density is near half, the embedding preserves (fully) the information of the SBM (i.e., the embedding is invertible). 

\textbf{Sparse/heterophilic communities $\Pi_s$ (Theorem 1.2):} When the two communities are \emph{sparse} relative to the cross-edges, the embedding becomes degenerate (maps all nodes to the zero vector, or similar). 
Intuitively, the cross-community (heterophilic) edges overpower the intra-community (homophilic) edges, so both the community structure and the edge density information is lost by the embedding. 

\textbf{Intermediate region $\Pi_m$ (Theorem 1.3):} Otherwise, the embedding preserves only partial information: the mapping of nodes to communities is encoded, but the densities of the communities are only encoded in a relative way. This lossy embedding also implies that, for each embedding, there is an equivalence class of SBMs that all map to it, which we discuss in Section~\ref{sec:eqclasses}. 


We formally state this result as follows. 

\begin{restatable}[Information Loss Regimes]{theorem}{embeddinglimitregimes}
Let $\mathcal G$ be a 2-block SBM($[p,q;q,r], (a, 1-a))$), and let $\omega_i$ be the embeddings learned by Algorithm 1 at convergence. Then, the kernel matrix $K(\mathcal G) \in \mathbb R^{n\times n}$ is a block-constant matrix: 
\[
K(\mathcal G) =
\left[\begin{array}{@{}l l@{}}
    K_1 \mathbf{1}_{an \times an} & K_2 \mathbf{1}_{an \times (1-a)n} \\[6pt]
    K_2 \mathbf{1}_{(1-a)n \times an} & K_3 \mathbf{1}_{(1-a)n \times (1-a)n}
\end{array}
\right],\]

where the values of $K_i \in \mathbb R$ depend on $a,p,q,r$ as follows:
\begin{enumerate}
    \item If $p, r \geq 1/2$ and $\sigma^{-1}(q)^2 < \sigma^{-1}(p) \sigma^{-1}(r)$, then  $K_1 = \sigma^{-1}(p), K_2 = \sigma^{-1}(q), K_3 = \sigma^{-1}(r)$. 

    \item Otherwise, if $p, r \leq \frac12$ and $(\frac12-q)^2 \leq (\frac12-p)(\frac12-r)$, then $K_1 = K_2 = K_3 = 0$. 

    \item Otherwise, 
    \[K_2 = \begin{cases} 
        -\sqrt{K_1K_3} & \text{if } q\leq \frac12 \\ 
        +\sqrt{K_1K_3} & \text{if } q>\frac12 \end{cases}, \] 
    with $K_1, K_3$ satisfying: 
    \begin{equation}
        \begin{split}
            a^2(\sigma(K_1)-p) - \mu K_3 &= 0  \\ 
            (1-a)^2(\sigma(K_3)-r) - \mu K_3 &= 0 \qquad  \mu \geq 0\\ 
            a(1-a)(\sigma(K_2) - q) + \mu K_2 &= 0 
        \end{split}
        \label{eq:intermediatecase}
    \end{equation}
\end{enumerate}
\label{thm:embedding-limit-regimes}
\end{restatable}

Proof in Appendix~\ref{appx:proof1}. A few notable observations: 

\begin{itemize}
    \item The regions of parameter space do \emph{not} depend on $a$, but only on $p,q,r$. That is, the information loss depends only on edge densities, not the size of the communities themselves. 
    \item Most ``real-world'' graphs fall in the intermediate region $\Pi_m$: they have both sparsity and homophily, with within-group edges typically denser than between-group edges. It is rare to have strongly heterophilic communities ($\Pi_s$) or extremely-dense graphs $(\Pi_d)$. 
    \item Community structure information is generally preserved in the embedding, even as degree/density information is not. 
    \item The sparse, information-loss region $\Pi_s$ is convex, while the dense, information-preserving region $\Pi_d$ is not (especially when $q$ is small). 
\end{itemize}

\subsection{Numerical characterizations for $k$-block SBMs}
\label{sec:numerical-1}

\begin{figure}
    \centering
    \includegraphics[width=\linewidth]{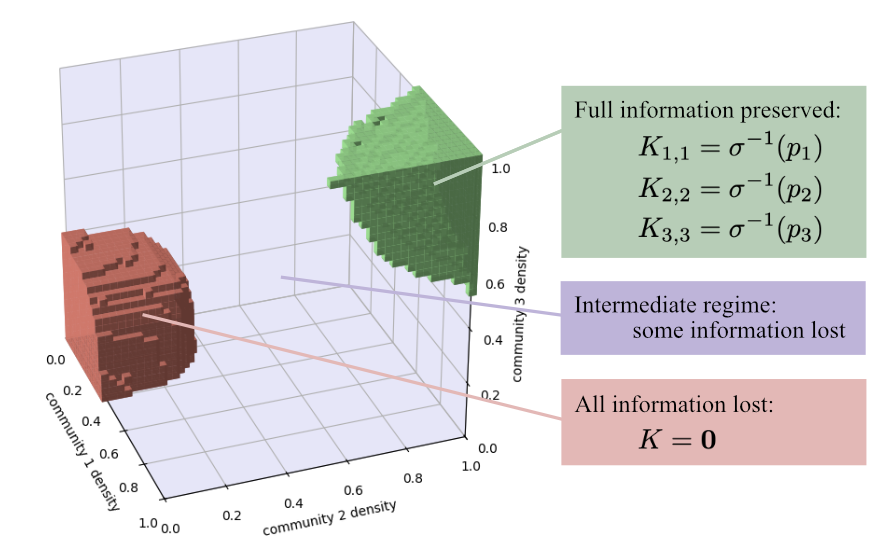}
    \caption{A 3-dimensional slice of the 6-dimensional parameter space of 3-block SBMs, showing SBMs with within-group edge probabilities of $(x,y,z)$ and between-group probability $0.3$. The regions of full, partial, and no information loss closely resemble the 2-block case.}
    \label{fig:numerical-regions}
\end{figure}

Finding an analogous exact characterization of information loss regimes in higher-order SBMs is difficult, largely due to the positive semi-definiteness constraint on $K$, the kernel matrix.\footnote{The proof of Theorem~\ref{thm:embedding-limit-regimes} relies on the algebraic inequalities equivalent to the PSD condition for $2\times 2$ matrices; similar inequalities exist for $3\times 3$ matrices, though they are substantially more complex. There are no algebraic inequalities equivalent to PSD for $4\times 4$ matrices and above.} Instead, we numerically evaluate minimizers of the empirical risk (Eq.~\ref{eq:empirical-risk}) for SBMs with more than two communities.

We find similar regions of full information preservation for dense communities, full information loss for sparse communities, and an intermediate region. The dense region, in particular, is characterized by the natural extension of the $\Pi_d$ condition from the 2-block case. 

\begin{empobservation}
\label{obs:PSD_condition}
A $k$-block SBM with probability matrix $P$ is embedded with no information loss whenever $\sigma^{-1}(P)$ (applied elementwise) is positive semi-definite.
\end{empobservation}


For the 3-community case (which has a 6-dimensional parameter space), we show a representative slice in Fig.~\ref{fig:numerical-regions}. In particular, the region boundaries closely resemble 3-dimensional analogues of the 2-dimensional boundary lines in Fig.~\ref{fig:regions-split}. These higher-order boundary regions have the same broad dependencies on ratios between within- and between-block probabilities, as well as convexity patterns. We describe these patterns and give further illustrative examples in Appendix~\ref{appx:numerical_infoloss}.

\section{Identically-embedded SBMs} 
\label{sec:eqclasses}
\begin{figure*}[!tbp]
    \begin{subfigure}[t]{0.3\linewidth}
        \vspace{0pt}
        \includegraphics[width=0.95\linewidth]{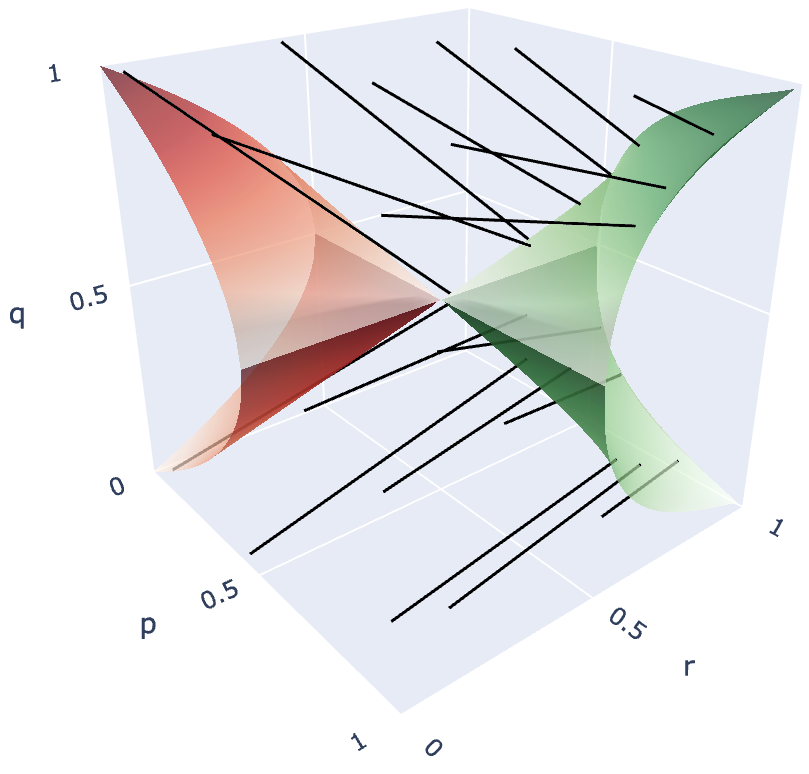}
        \caption{The equivalence classes of identically-embedded graphs are lines within $\Pi_m$. The slope of each line is controlled by the parameter $\eta$.}
        \label{fig:lines-of-eq-classes}
    \end{subfigure}
    \hfill 
    \begin{subfigure}[t]{0.68\linewidth}
        \vspace{0pt}
        \includegraphics[width=\linewidth]{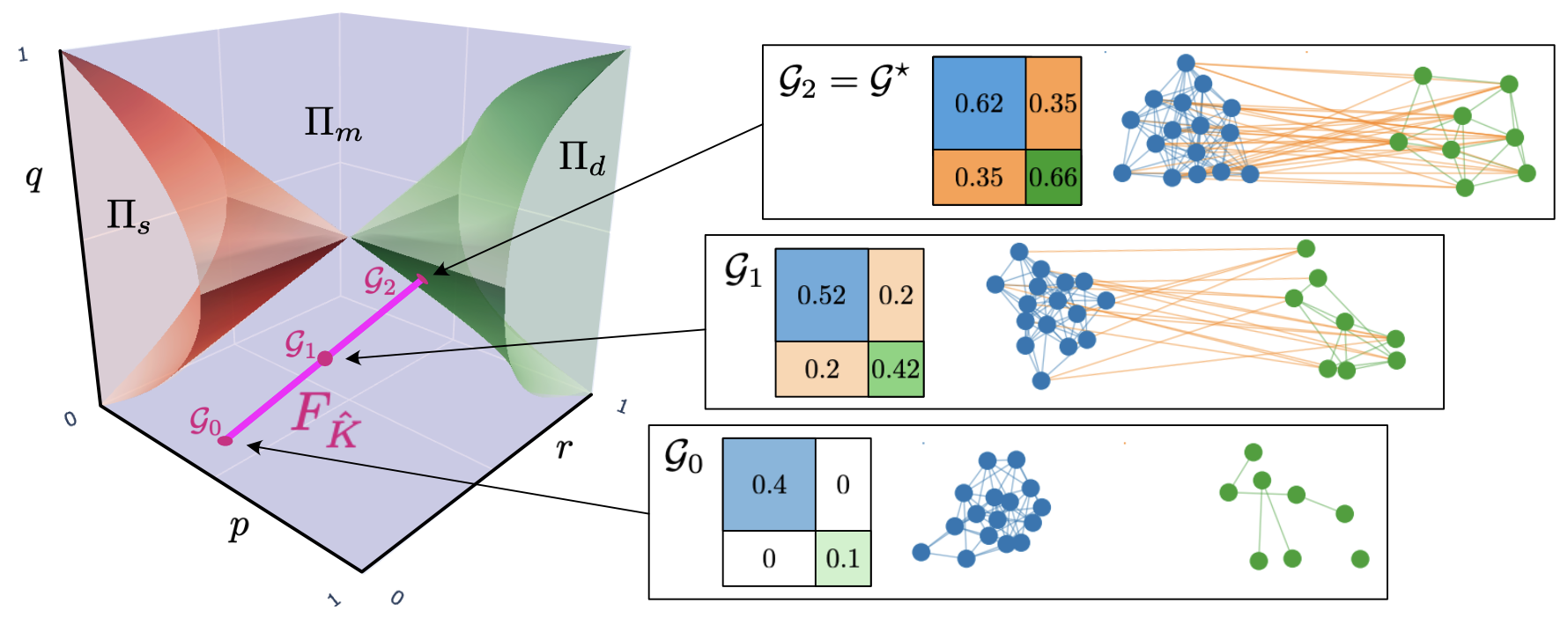}  
        \caption{An example of one such equivalence class, for the embedding $\hat K = (0.49, -0.59, 0.71)$ with $a=0.66$, resulting in $\eta=0.62$. All points along the {\color{magenta} $F_{\hat K}$ line} have the same embedding, including graphs that are both maximally-dense ($\mathcal G^\star$, on the boundary of $\Pi_d$) and maximally-sparse ($\mathcal G_0$).}
        \label{fig:one-eqclass}
    \end{subfigure}
    \caption{Examples of equivalence classes in $(p,q,r)$ space for $a=0.66$.}
\end{figure*}

Most realistic networks lie in the intermediate regime $\Pi_m$, where some (but not all) information is lost. This information loss induces \emph{equivalence classes} of input graphs, where two graphs are in the same class if they receive the same embedding. What structure do these equivalence classes have? For example, one might hope that only very similar graphs receive identical embeddings: this would make this information loss behave like local ``noise,'' but leave large differences largely untouched. 
Unfortunately, this is not the case, as we show in this section; each equivalence class contains vastly different graphs. 


\subsection{Characterizing equivalence classes of 2-block SBMs}

As before, we begin with an analytic characterization for 2-block SBMs. We find that each equivalence class consists of a line in ($p,q,r$) parameter space, ranging from maximally dense communities (on the boundary of $\Pi_d$) to maximally sparse ones (when $q=0$ or $q=1$). This is quite a large range: the overall density of the graph is \emph{not} encoded in the representation! That is, the embedding of a very sparse graph can be identical to that of a very dense graph. 

We explicitly characterize the form of an equivalence class induced by a given embedding $K$:

\begin{restatable}[Identically-embedded 2-block SBMs]{theorem}{linearfamilies} Partition $\Pi_m$ into equivalence classes, where  $\mathcal G \sim \mathcal G'$ iff $K(\mathcal G) = K(\mathcal G')$. Fix any $\mathcal G = (p_0,q_0,r_0)$ and let $\hat K = K(\mathcal G) = (K_1, K_2, K_3)$. The equivalence class containing $\mathcal G$ is 
\begin{equation}
F_{\hat K} = \left\{ (p,q,r) \in \Pi_m, \text{ where }\; \begin{aligned}
   p &= p_0 + \eta \Delta \\
   q &= q_0 + \mathbb S \Delta \\
   r &= r_0 + \tfrac{1}{\eta}\Delta
\end{aligned}, \Delta \in \mathbb R \right\},
\label{eq:linearfam_1}
\end{equation}
where 
\[ \eta = \frac{1-a}{a} \sqrt{\frac{K_3}{K_1}} \quad \text{ and }\quad  \mathbb S = \begin{cases} +1 & q\leq 1/2 \\ -1 & q > 1/2 \end{cases}.\]

Equivalently, we can write the class as
\begin{equation}
\begin{aligned}
    F_{\hat K} = \bigg\{\mathcal G' = \left(\begin{bmatrix}
        p & q \\ q & r
    \end{bmatrix} + \mathbb S \Delta \begin{bmatrix}
        \eta & 1 \\ 1 & 1/\eta
    \end{bmatrix}\right)\\ 
    \text{where } \mathcal G' \in \Pi_m, \Delta \in \mathbb R \bigg\}.
\end{aligned}
\end{equation}
\label{thm:eqclasses}
\end{restatable}

Proof in Appendix~\ref{appx:proof2}.

Each equivalence class $F_{\hat K}$ is a line in parameter space with ``slope'' determined by $\eta$. For a sparse graph ($q \leq \frac12$), one can imagine the class as a series of increasingly-dense SBMs, where $p$ increases at rate $\eta$, $r$ increases at rate $1/\eta$, and $q$ increases at rate 1. A visualization of several of these linear equivalence classes is shown in Fig.~\ref{fig:lines-of-eq-classes}. Each class is a line from $q=0$ or $q=1$ to the surface of $\Pi_d$. 

One particular class is shown in Fig.~\ref{fig:one-eqclass}. Note that even the relative densities of the communities are not preserved by the embedding: in $\mathcal G_0$, the larger (blue) community is denser, whereas in $\mathcal G_2$ the smaller (green) community is denser. In fact, we can exactly derive what the densest member of each equivalence class is:

\begin{corollary}
    For any SBM $\mathcal G \in \Pi_m$ with embedding $\hat K=(K_1,\pm \sqrt{K_1 K_3},K_3)$, the densest SBM $\mathcal G^\star$ (in $p,r$) that has the same embedding is 
    \[(p^\star, q^\star, r^\star) \triangleq \argmax_{\mathcal G' : K(\mathcal G') = K(\mathcal G)} p' =  (\sigma(K_1), \sigma(K_2), \sigma(K_3)).\] 
    This SBM $(p^\star, q^\star, r^\star)$ lies on the boundary between $\Pi_m$ and $\Pi_d$. 
    \label{cor:densecanonical}
\end{corollary}

It is easy to verify that this choice of $(p^\star, q^\star, r^\star)$ both matches the $K$ values of Theorem~\ref{thm:embedding-limit-regimes}.1 and satisfies Eq.~\ref{eq:intermediatecase}.
Cor.~\ref{cor:densecanonical} also allows us to restate $F_{\hat K}$ without relying on a prior-known $(p_0, q_0, r_0)$, by using $(p^\star, q^\star, r^\star)$ instead:

\begin{equation}
F_{\hat K} = \left\{ (p,q,r) \in \Pi_m, \text{ where }\; \begin{aligned}
   p &= \sigma(K_1) + \eta \Delta \\
   q &= \sigma(\mathbb S \sqrt{K_1K_3}) + \mathbb S \Delta \\
   r &= \sigma(K_3) + \tfrac{1}{\eta}\Delta\\
    & \text{for } \Delta \in \mathbb R
\end{aligned}\right\}
\label{eq:linearfam_1}
\end{equation}

\subsection{Numerically computing equivalence classes for $k$-community SBMs}

Once again, we turn to numerical methods to study higher-order SBMs. 
Using a numerical Jacobian on the minimizer described in \ref{sec:numerical-1} plus a corrector, we are able to step along the equivalence class given an initial (graph, embedding) pair. We find, again, that the equivalence classes are linear, with the slope of each line determined by a vector of relative densification rates for each block, corresponding to $\eta$ in the 2-block case. As before, each class ranges from maximally-sparse to maximally-dense: one endpoint is a SBM with at least one zero component, and the other lies on the boundary of the dense region (with block probabilities $P = \sigma(K)$, applied elementwise). Additionally, we find that the ``densification rate'' has low-rank-like behavior: for details, see Appendix~\ref{appx:numericaleqclasses}.

Because these classes live in a space whose dimension grows quadratically with the number of communities, they are difficult to visualize. 
We give two illustrative examples of equivalence classes for 3-block SBMs in Appendix~\ref{appx:numericaleqclasses_examples}.

\section{Fairness and downstream tasks}
\subsection{Representation-level disparities}
\begin{figure*}
\begin{subfigure}[t]{0.48\linewidth} 
    \vspace{0pt}
    \centering
    \includegraphics[width=\linewidth]{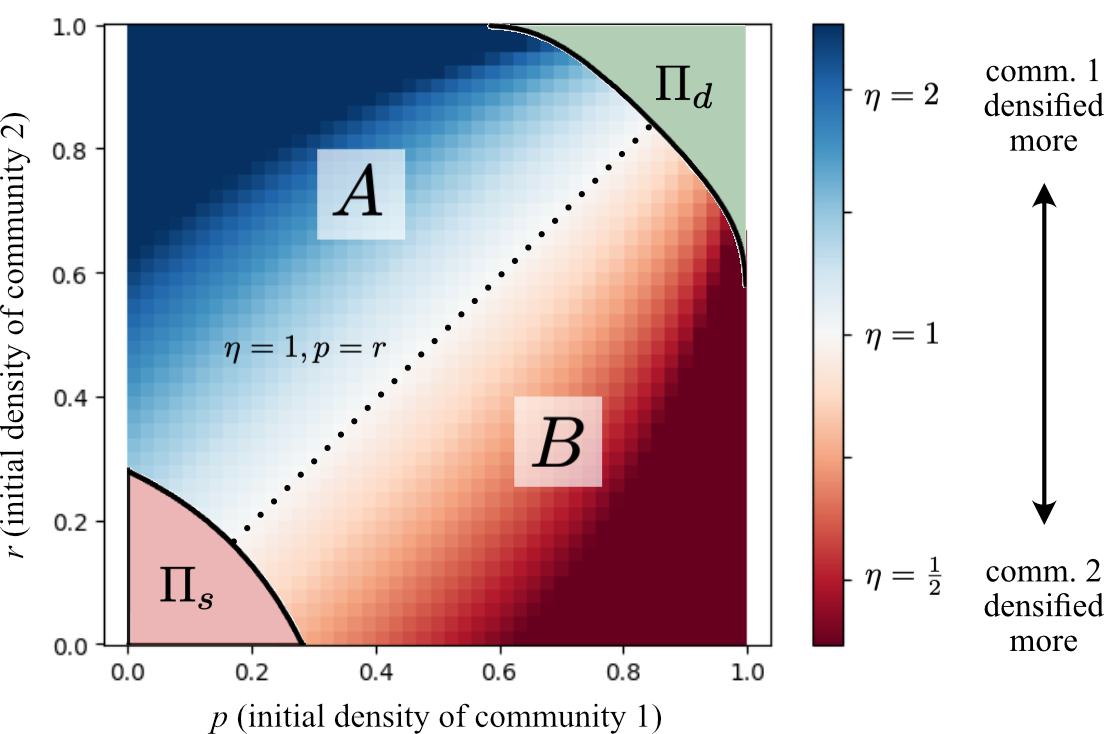}
    \caption{
    When the two communities are equal-sized $(a=0.5)$, link prediction strengthens the sparser community. In region $A$, community 1 is sparser ($p<r$) and is densified (because $\eta > 1$); the same is true for community 2 in region $B$.}
\end{subfigure}
\hfill 
\begin{subfigure}[t]{0.48\linewidth} 
    \vspace{0pt}
    \centering
    \includegraphics[width=\linewidth]{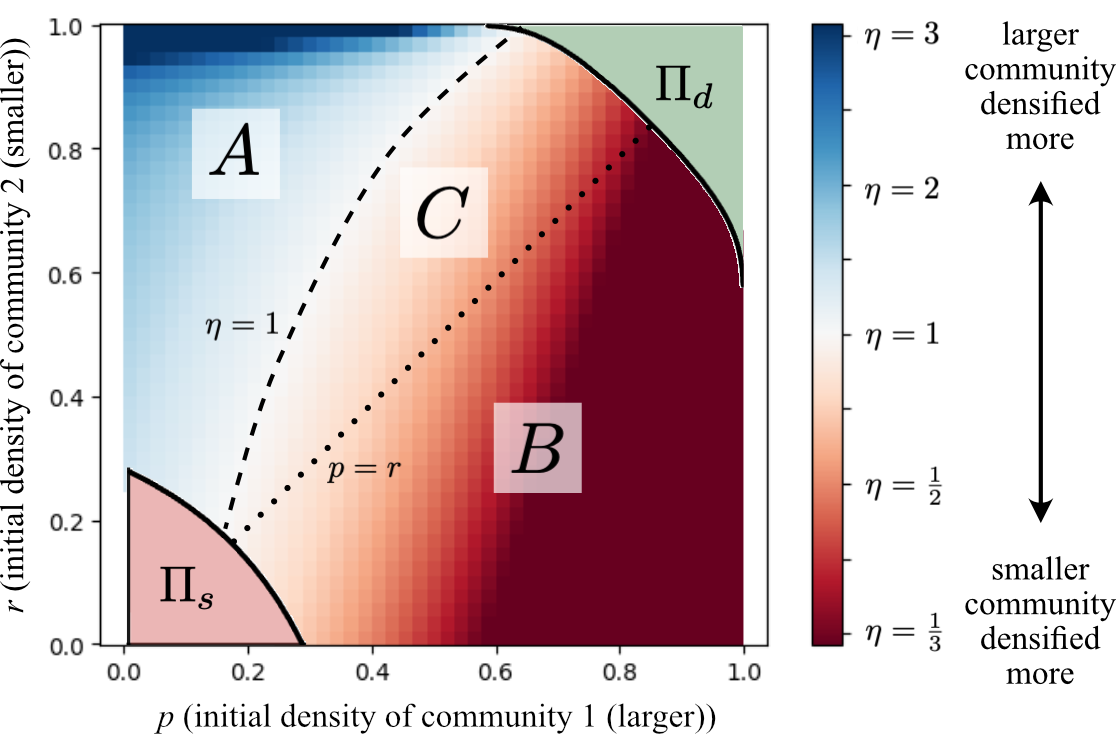}
    \caption{When the two communities are different sizes $(a=0.75)$, link prediction is also pushed towards strengthening smaller communities. In region $C$, this force ``wins out'' over the sparsity factor: community 2 (which is $\frac14$ of the graph) is strengthened despite being denser.} 
\end{subfigure}
\caption{Unbalanced densification in link prediction is affected by the relative sparsity and size of communities in the graph.}
\label{fig:etas-unbalanced}
\end{figure*}

One takeaway from our results about equivalence classes is that, when we embed a sparse graph, we get back an embedding that ``looks like'' a much denser graph\footnote{And, similarly, when we embed a dense graph, we get back an embedding that looks like a much sparser one. But most real-world networks are very sparse, so we'll use this language for now.}. But not every community is densified the same amount! The larger community is densified at rate $\eta = \frac{1-a}{a}\sqrt{\frac{K_3}{K_1}}$, and the smaller at rate $\frac1\eta = \frac{a}{1-a} \sqrt{\frac{K_1}{K_3}}$. 

This directly explains the ``squishing'' behavior we see in our motivating examples, in which smaller and sparser communities were compressed into smaller (and denser) sections of embedding space. $\eta$ depends on two ratios:
\begin{itemize}
    \item $\frac{1-a}{a}$, the ratio of the group sizes (larger difference in group sizes leads to smaller $\eta$, meaning the small group is densified); and 
    \item $\sqrt{\frac{K_3}{K_1}}$, the ratio of the within-group inner products in the learned embedding. Recall that the loss drives $K_3$ towards $\sigma(r)$, and $K_1$ towards $\sigma(p)$. Since $\sigma$ is monotonic, this term increases in $p/r$, the ratio between the groups' densities\footnote{except near the boundary of $\Pi_s$, the sparse, heterophilic regime}. Once again, a larger difference in density translates to an increased densification of the sparser group. 
\end{itemize}

In real networks, these two pressures are often aligned: smaller communities also tend to be sparser. However, they can also be in tension; in these cases, which one ``wins out'' depends on the relative strength of each. We show two slices of parameter space, with $a=0.5$ (equal-sized communities) and $a=0.75$ (one much larger community), in Fig.~\ref{fig:etas-unbalanced}. 

How does this disparate representation affect downstream tasks? 
In general, the information loss described in Theorem~\ref{thm:embedding-limit-regimes} gives reason to believe that using the embeddings alone for downstream tasks is often insufficient, particularly in the non-invertible regime $\Pi_m$. Because a particular embedding $K$ can come from a whole class of graphs, many downstream tasks may essentially boil down to ``guessing'' which graph in the class the embedding is from. 

\subsection{Link prediction}
\label{sec:link-prediction}

One common task is link prediction: given the embeddings $\{\omega_0, \omega_1, \cdots\}$, output a probability for each potential edge $(i,j)$. \textbf{A baseline link-prediction algorithm predicts the densest graph in the equivalence class}. A simple algorithm is to compute the pairwise inner products between the embeddings, and then pass these scores through a sigmoid such as $\sigma$, the standard logistic\footnote{This method is used in variational graph auto-encoders \cite{kipf2016variational} and is representative of a whole family of inner-product-based link prediction algorithms (e.g., \cite{saxena2022nodesim, malitesta2024dot}).}. The predicted probability of edge $(i,j)$ in the 2-block case is
\begin{align*}
    \sigma(\langle \omega_i, \omega_j\rangle) &= \sigma(K[i,j]) \\ 
    &= \begin{cases} \sigma(K_1) = p^\star & \text{ if }i,j \text{ in group 1}\\\sigma(K_2) = q^\star & \text{ if }i,j \text{ in different groups}\\\sigma(K_3) = r^\star & \text{ if }i,j \text{ in group 2}\end{cases}
\end{align*}

That is, using this algorithm implicitly recovers the maximally-dense SBM $\mathcal G^\star = (p^\star, q^\star, r^\star)$ in the equivalence class (c.f. Cor.~\ref{cor:densecanonical}). In particular, this will \textbf{overpredict edges for smaller and sparser communities} due to the unbalanced densification via $\eta$. Erroneous overpredicted edges increase the error rate of these minority communities, leading to the aforementioned violations of statistical parity and group fairness. 

In Appendix~\ref{appx:fb100}, we embed graphs from the Facebook100 university dataset, to demonstrate that these effects are not just an artifact of stochastic block models. In particular, we find that the patterns of sparser, smaller communities getting over-densified is robustly observed in real graphs, and that the empirical rate of densification is strongly correlated with $\eta$, the amount predicted by our theory.



\subsection{Other downstream tasks}
The non-invertibility and relative densification we observe have implications for other downstream tasks. We briefly discuss community detection in Appendix~\ref{appx:tasks}.


\subsection{Fairness}
Substantial attention has been paid to whether network algorithms (such as link prediction) exacerbate or mitigate underlying biases in network structure (see \cite{saxena2022nodesim} for an overview). 

It is common to consider the density of a community as a measure of the ``strength'' or ``power'' of that community within a network \cite{freeman1978centrality}; so, depending on the context, it can be considered preferable to preferentially strengthen sparser and/or smaller communities \cite{stoica2024fairness}, to equalize link probabilities relative to group membership  \cite{liu2024promoting, li2022fairlp}, or to reinforce between-community edges \cite{conover2011political}\footnote{This is often proposed in the context of polarization, although the empirical support for this is mixed \cite{dandekar2013biased,bakshy2015exposure}.}. In our case, the disproportionate densification of smaller and sparser communities can be viewed as \emph{reducing} (and in some cases, even reversing) degree inequality. Depending on one's goals, this may be seen as beneficial (mitigating structural biases) or harmful (introducing group-level error). 

More broadly, our non-invertibility results illustrate a general point about downstream tasks: by making algorithmic choices (e.g., choice of parameters, normalizations, assumptions, etc.), one is \textbf{implicitly choosing a ``canonical'' member among equivalently-embedded graphs}. We emphasize that this choice should be made thoughtfully and not in a data-agnostic way. 


\section{Discussion} 
The theoretical conditions we derive in this work allow us to draw principled connections between graph structure, information loss, and learned embeddings.  We are excited that patterns predicted by our theoretical results contribute to explaining disparities in existing ``baseline'' algorithms. 



While the stochastic block model we study captures many aspects of real-world graphs,
there are many other properties 
that they do not capture. We expect that that our results can be generalized in multiple directions, including to more complex graphons and to more general generative models (e.g., graphexes \cite{veitch2015class}). 
Another open question is whether these results transfer to ``word-context'' style embeddings as described by Chanpuriya et al \yrcite{chanpuriya2020node}, which have stronger representation power. 
Additionally, Davison \& Austern \yrcite{davison2023asymptotics} characterize limiting distributions of loss-minimizing embeddings that use more complex sampling schemes (such as the random walks used by DeepWalk and node2vec); generalizing our results to these schemes is a direction for future work. 




\section*{Impact Statement}

Our work addresses fairness in social networks by analyzing representation disparities across majority and minority groups. By furthering theoretical understanding of the sources of these disparities, we offer insights that can inform principled algorithmic interventions, helping future systems achieve more equitable outcomes.


\bibliographystyle{icml2026}
\bibliography{bibliography}

\newpage 
\onecolumn 
\appendix 

\section{Proofs for 2-Block SBMs}
\subsection{Theorem~\ref{thm:embedding-limit-regimes}: Information Loss Regimes}
\label{appx:proof1}


\embeddinglimitregimes*

\begin{proof}
    Let $\rho$ be the sparsity parameter of the uniform vertex sampling, and our two-community SBM have parameters $(a,(p,q,r))$. Let $K$ be the matrix of ERM-minimizing inner products (that is, $K_{i,j} = \langle \hat \omega_i, \hat \omega_j\rangle$). By Davison \& Austern \yrcite{davison2023asymptotics}, $K$ converges to a positive definite block matrix 
    \[ K = \begin{bmatrix}
        K_1 & K_2 \\ K_2 & K_3
    \end{bmatrix}\]
    where the block sizes are $an, (1-a)n$, corresponding to the sizes of the SBM communities. We wish to find the values of $K_1, K_2, K_3$ that minimize the empirical risk (Eq.~\ref{eq:empirical-risk}): 
    \begin{align*}
        & \mathcal{R}_n(\omega_1, \cdots, \omega_n) \\
        &= \sum_{i,j \in [n], i\neq j} \Pr((i,j) \in S(G_n)\mid G_n) \ell(\langle\omega_i, \omega_j\rangle, a_{ij}) \\ 
        &= \frac{\rho}{n^2} \sum_{i \neq j} \ell(\langle \omega_i, \omega_j \rangle, a_{ij} ) \qquad \text{(uniform vertex sampling)}
    \end{align*} 
    
    Denoting the two SBM communities by $A,B$, we split the sum into cases: 
    \begin{align*}
        &= \frac{\rho}{n^2}\bigg( \sum_{i,j\in A} \ell(\langle w_i, w_j\rangle, p) + \sum_{i,j \in B} \ell(\langle w_i, w_j\rangle , r) + 2\sum_{i \in A, j \in B}\ell(\langle w_i, w_j\rangle, q)  \bigg)\\
        &= \rho\left(a^2 \ell(K_1, p) + (1-a)^2 \ell(K_3, r)+ 2a(1-a)\ell(K_3, q) \right)
    \end{align*}

    Letting $\sigma(y) = \frac{e^y}{1+e^y}$, the cross-entropy loss is
    \begin{align*} \ell(y,x) &= -x\log(\sigma(y)) - (1-x) \log (1-\sigma(y)) \\
    &= \log(1+e^y) - xy 
    \end{align*}
    So, the problem of optimizing the risk $\mathcal R_n$ is equivalent to minimizing the function 
    \begin{align*}
    \mathcal R_n(K) = & \rho\bigg( a^2(\log(1+e^p)-K_1p)  + (1-a)^2 (\log(1+e^r)-K_3r) + 2a(1-a) (\log(1+e^q) - K_2q) \bigg)    
    \end{align*}
    
    over positive definite matrices $K = [K_1, K_2; K_2, K_3]$. The positive definite constraint forces $K_1, K_3 \geq 0$ and $K_2^2 \leq K_1K_3$. Letting $\mu_1, \mu_2, \mu_3 \geq 0$ be dual variables, the Lagrangian is $\mathcal R_n(K) - \mu_1 K_1 - \mu_2 K_3 - \mu_3(K_1K_3-K_2^2)$. The KKT conditions for this problem state that any minima that satisfies the LICQ conditions must satisfy 
    \begin{align*}
           a^2(\sigma(K_1)-p) - \mu_1 - K_3\mu_3 &= 0, & (1)\\ 
           (1-a)^2(\sigma(K_3)-r) - \mu_2 - K_1\mu_3 &= 0, & (2) \\ 
           2a(1-a)(\sigma(K_2)-q) + 2\mu_3K_2 &= 0, & (3)\\
           \mu_1K_1 = 0, \qquad \mu_2K_3 = 0, \qquad \mu_3(K_1K_3-K_2^2) &= 0.
    \end{align*}

    We now work case-by-case on which constraints are tight (i.e., whether $\mu_1, \mu_2, \mu_3$ are zero). Each case is only feasible in a particular region of $(p,q,r)$ parameter space; we show that the cases are feasible in regions of parameter space that are \emph{disjoint} (except potentially at the boundaries). This allows us to conclude that the values of $K$ that are taken in each case are the optimal ones.

    \vspace{10pt}

    \textbf{Case 1: No constraints are tight (dense case)}. That is, $K_1 > 0, K_3 > 0, K_2^2 < K_1K_3$. Then $\mu_1 = \mu_2 = \mu_3 = 0$. (1) reduces to $a^2(\sigma(K_1)-p) = 0 \implies K_1 = \sigma^{-1}(p)$. Similarly, $K_2 = \sigma^{-1}(q)$ and $K_3 = \sigma^{-1}(r)$. Since $K_1, K_2, K_3$ are positive, this exists only when $p >\frac12, r>\frac12$. The $K_2^2 < K_1K_3$ constraint implies $\sinv(q)^2 < \sinv(p)\sinv(r)$. 

    So, this case is feasible so long as $p>\frac12, r>\frac12, \text{ and }\sinv(q)^2 < \sinv(p)\sinv(r)$ and results in $K_1 = \sinv(p), K_2 = \sinv(q), K_3 = \sinv(r)$.

    \vspace{10pt}

    \textbf{Case 2: $K_2^2 \leq K_1K_3$ is tight, and the other constraints are loose (middle case).} So, $K_1>0, K_3>0$; these imply $\mu_1 =\mu_2 = 0$ and $\mu_3 \geq 0$. Here, we split further into two cases: 
    \vspace{10pt}

    \textbf{Case 2(a): $K_2 = -\sqrt{K_1K_3}, K_1>0, K_3>0$ (sparse cross-edges)}. The KKT conditions become: 
    \begin{align*}
           a^2(\sigma(K_1)-p) - K_3\mu_3 &= 0, & (4)\\ 
           (1-a)^2(\sigma(K_3)-r) - K_1\mu_3 &= 0, & (5) \\ 
           2a(1-a)(\sigma(-\sqrt{K_1K_3})-q) - 2\mu_3\sqrt{K_1K_3} &= 0, & (6)
    \end{align*}

    The first equation implies $a^2(\sigma(K_1)-p) \geq 0 \implies \sigma(K_1) \geq p$. Combined with the constraint that $K_1 \geq 0$, then $K_1 \geq \max(0, \sinv(p))$. Similarly, $K_3 \geq \max(0, \sinv(r))$; so, $\sqrt{K_1K_3} \geq \sqrt{\max(0, \sinv(p)) \max(0, \sinv(r))}$. 

    Rearranging the third equation, we know 
    \[ \sigma(-\sqrt{K_1K_3}) - q = \frac{\mu_3}{a(1-a)} \sqrt{K_1K_3} \geq 0\]
    So $\sigma(-\sqrt{K_1K_3}) \geq q$, which implies $\sqrt{K_1K_3} \leq -\sinv(q)$. 

    Since $\sqrt{K_1K_3} > 0$, this is only feasible when $-\sinv(q) >0$, so $q<\frac12$. 

    Finally, we combine the bounds on $\sqrt{K_1K_3}$: 
    \begin{align*}
        &\sqrt{\max(0, \sinv(p)) \max(0, \sinv(r))} \leq \sqrt{K_1K_3} \leq -\sinv(q)\\ 
        \implies & \sinv(q)^2 \geq \max(0, \sinv(p)) \max(0, \sinv(r)) \\ 
        \equiv \quad & \sinv(q)^2 \geq \sinv(p) \sinv(r) \;\text{ or } \;p \leq \frac12 \;\text{ or } \; r \leq \frac12
    \end{align*}

    So this case ($K_2=-\sqrt{K_1K_3}, K_1>0, K_3>0$) is feasible so long as $q <\frac12$ and ($\sinv(q)^2 \geq \sinv(p) \sinv(r)$ or $p \leq \frac12$ or $r \leq \frac12$). 

    \vspace{10pt}
    \textbf{Case 2(b): $K_2 = +\sqrt{K_1K_3}, K_1>0, K_3>0$ (dense cross-edges)}. This case is very similar to case 2(a): $\sqrt{K_1K_3} \geq \sqrt{\max(0,\sinv(p))\max(0, \sinv(r))}$. However, the sign changes means that $\sqrt{K_1K_3} \leq \sinv(q)$, which implies $q < \frac12$.

    Squaring both sides as before, we obtain the same condition: 
    \[\sinv(q)^2 \geq \sinv(p)\sinv(r) \;\text{ or } \;p \leq \frac12 \;\text{ or } \; r \leq \frac12\]

    So this case ($K_2=+\sqrt{K_1K_3}, K_1>0, K_3>0$) is feasible so long as $q >\frac12$ and ($\sinv(q)^2 \geq \sinv(p) \sinv(r)$ or $p \leq \frac12$ or $r \leq \frac12$). 

    \vspace{10pt}
    \textbf{Case 3: $K_1=0$ or $K_3 =0$ or both (edge cases)}. In all of these cases, the PSD constraint $K_2^2 \leq K_1K_3$ forces $K_2 = 0$. Then, the third KKT condition ($a(1-a)(\sigma(0)-q) - 2\mu_3(0) =0$) forces $q=\frac12$. 
    \begin{itemize}
        \item When $K_1 = 0, K_3 \neq 0$, then $\mu_2 = 0$. The first KKT condition $\mu_1 \geq 0$ forces $a^2(\sigma(0)-p)\geq 0$, that is, $\frac12 \geq p$. The second KKT condition becomes $(1-a)^2(\sigma(K_3)-r) - 0 - 0\mu_3$, which is satisfied only when $K_3 = \sinv(r)$. So, the feasible region is $q=\frac12, p \leq \frac12, r > \frac12$.
        \item Similarly, when $K_3 = 0, K_1\neq 0$, the second KKT condition forces $\frac12 \geq r$, and the first KKT condition forces $K_1 = \sinv(p)$. So, the feasible region is $q=\frac12, r \leq \frac12, p > \frac12$. 
        \item When both $K_1 = K_3 = 0$, then the feasible region is $q=\frac12, p \leq \frac12$ and $r\leq \frac12$. 
    \end{itemize}

    \vspace{10pt} 
    \textbf{Non-KKT cases}. 
    Finally, we check points that \emph{don't} satisfy the linear independence constaint qualification (LICQ) conditions, i.e., potential solutions that would not be found by KKT. Our constraints are $K_1 \geq 0, K_3 \geq 0, K_1K_3-K_2^2 \geq 0$. Their gradients are $(1,0,0), (0,0,1)$, and $(K_3, -2K_2, K_1)$ respectively. These gradients are linearly independent at all points except for $(0,0,0)$ (the origin) and $(0, 0, K_3)$ and $(K_1,0,0)$. We check the optimality of these points first: 
    \begin{itemize}
        \item $(0,0,0)$ is globally optimal iff the directional derivative in every feasible direction is nonnegative. That is, for all feasible $(u,v,w)$ in the feasible cone (that is, $u \geq 0, w \geq 0, v^2 \leq uw)$, \[\nabla \mathcal R_n(0,0,0) \cdot (u,v,w) \geq 0.\]
        $\nabla \mathcal R_n(0,0,0)$ is $(a^2(1/2-p), 2a(1-a)(1/2-q), (1-a)^2)1/2-r)$. Letting $u = x^2, w = z^2, v = \rho xz$ for $x \geq 0, z\geq 0, |\rho|\leq 1$ (which ensures the constraints are satisfied), this is 
        \begin{align*}
            &\nabla \mathcal R_n(0,0,0) \cdot (u,v,w) \\ &= \left(\frac12-p\right)(ax)^2 + 2\left(\frac12-q\right)\rho (ax)(1-a)z + \left(\frac12-r\right) ((1-a)z)^2 \geq 0
        \end{align*} 
        This is minimized by having $\rho = \text{sign}(1/2-q)$. Rearranging and letting $t=\frac{ax}{(1-a)z}$, this is
        \[ \frac{1}{(1-a)^2z^2} \left(\frac12-p\right)t^2 + 2\left|\frac12-q\right| t + \left(\frac12-r\right) \geq 0\]
        which holds whenever $\frac12-p \geq 0, \frac12-r \geq 0$, and $(\frac12-q)^2 \leq (\frac12-p) (\frac12-r)$. 

        That is, when $p \leq \frac12, r\leq \frac12$, and $(\frac12-q)^2 \leq (\frac12-p) (\frac12-r)$, the risk is minimized by letting $K_1 = K_2 = K_3 = 0$. 

        \item $(K_1,0,0)$ is feasible along the curve $\gamma(t) = (K_1, t, t^2/K_1)$ for $t \in \mathbb R$. The directional derivative along $\gamma$ at $t=0$ is \[ \frac{d}{dt} \mathbb R_n(\gamma(t)) \mid_{t=0} = 2a(1-a)(1/2 - q).\]
        This is only minimizing at $q=\frac12$. This becomes identical to case 3 in the KKT cases. 
        \item $(0,0,K_3)$ is similarly identical to case 3 of the KKT cases. 
    \end{itemize}

    These are the only cases where a non-KKT case can be optimal. Since the cases are feasible in disjoint regions of parameter space, the theorem is proved. 
\end{proof}

\subsection{Theorem~\ref{thm:eqclasses}: Equivalence Classes in $\Pi_m$}
\label{appx:proof2}

\linearfamilies*

\begin{proof}

We recall the KKT-derived conditions from Theorem~\ref{thm:embedding-limit-regimes} (case 3): 

\begin{equation}
    \begin{split}
        a^2(\sigma(K_1)-p) - \mu K_3 &= 0  \\ 
        (1-a)^2(\sigma(K_3)-r) - \mu K_1 &= 0 \qquad \qquad  \qquad  \mu \geq 0\\ 
        a(1-a)(\sigma(K_2) - q) + \mu K_2 &= 0 
    \end{split}
    \tag{2}
\end{equation}

Fix $K_1, K_3$. In this case, 
\[K_2 = \begin{cases}
    -\sqrt{K_1K_3} & q \leq \frac12 \\ 
    +\sqrt{K_1K_3} & q > \frac12
\end{cases} \quad = \quad \mathbb S \sqrt{K_1K_3} \qquad\] 
where \[\mathbb S = \begin{cases} 1 & q\leq \frac12 \\ -1 & q > \frac12 \end{cases}\]

Suppose both $(p_0, q_0, r_0)$ and $(p_0+x, q_0+y, r_0+z)$ have the same embedding $\hat K = (K_1,\mathbb S K_2, K_3)$: that is, they result in solutions to  (\ref{eq:intermediatecase}) of $(K_1,K_3,\mu)$ and ($K_1,K_3,\mu')$ respectively. Then 
\begin{align*}
    &a^2(\sigma(K_1) - (p_0 + x)) - \mu' K_3 \\ 
    = \;& a^2 (\sigma(K_1) - p_0) - \mu K_3 = 0 \\ 
    \implies &-a^2x - \mu' K_3 = \mu K_3 \\ 
    \implies &x = (\mu-\mu') \frac{K_3}{a^2}
\end{align*}
Similarly, 
\[y = (\mu-\mu') \frac{\mathbb S \sqrt{K_1K_3}}{a(1-a)} \qquad \text{and} \qquad z = (\mu-\mu') \frac{K_1}{(1-a)^2} \]

So, 
\[x = y \cdot \mathbb S \frac{1-a}{a} \sqrt{\frac{K_3}{K_1}} \qquad \text{and} \qquad z = y \cdot \mathbb S \frac{a}{1-a} \sqrt{\frac{K_1}{K_3}}\]

Since $\mu'$ is free, the class is simply $\{(p_0+x, q_0+y, z_0+z) : \mu' \in \mathbb R\} \cap \Pi_m$. Letting $\eta = \frac{1-a}{a} \sqrt\frac{K_3}{K_1}$ and reparametrizing $\mathbb S y (\mu-\mu') =\Delta$, the class becomes 

\begin{equation}
F_{\hat K} = \left\{ (p,q,r) \in \Pi_m, \text{ where }\; \begin{aligned}
   p &= p_0 + \eta \Delta \\
   q &= q_0 + \mathbb S \Delta \\
   r &= r_0 + \tfrac{1}{\eta}\Delta
\end{aligned}, \Delta \in \mathbb R \right\},
\end{equation}    
\end{proof}

\begin{observation}
    We can also express 
    \[\eta = \sqrt{\frac{\sigma(K_1)-p}{\sigma(K_3)-r}}.\]This can be useful to avoid numerical issues with $\sqrt{\frac{K_1}{K_3}}$ when $K_1$ or $K_3$ are near 0, and also gives some intuition for an alternative interpretation of $\eta$: as a ratio of the amount of ``forced error'' between $\sigma(K_1)$ and $p$ and between $\sigma(K_3)$ and $r$. 
\end{observation}

\begin{proof}
    Rewriting the first two KKT conditions, 
    \[\mu K_3 = a^2(\sigma(K_1)-p) \qquad \mu K_1 = (1-a)^2 (\sigma(K_3)-r)\]
    Dividing these two equations: 
    \[ \frac{\mu K_3}{\mu K_1} = \frac{a^2(\sigma(K_1)-p)}{(1-a)^2 (\sigma(K_3)-r)}\]
    So,
    \begin{align*}
    \eta &= \frac{1-a}{a} \sqrt\frac{K_1}{K_3} = \frac{1-a}{a} \sqrt{\frac{a^2(\sigma(K_1)-p)}{(1-a)^2 (\sigma(K_3)-r)}} =  \sqrt{\frac{\sigma(K_1)-p}{\sigma(K_3)-r}}
    \end{align*}
\end{proof}
\newpage 
\section{Numerical Results for $k$-Block SBMs}

\subsection{Regimes of Information Loss}
\label{appx:numerical_infoloss}
\begin{figure}
    \centering
    \includegraphics[width=0.3\textwidth]{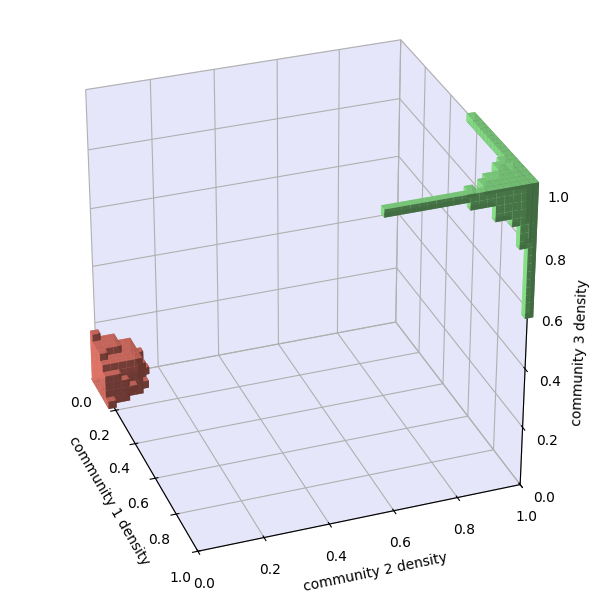}
    \includegraphics[width=0.3\textwidth]{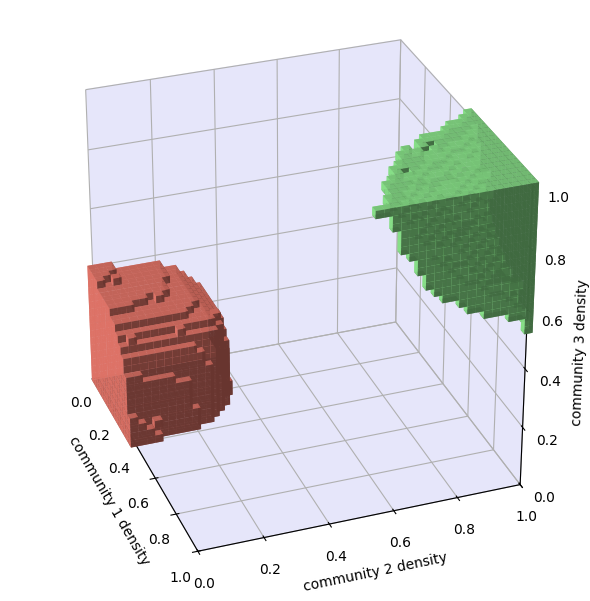}
    \includegraphics[width=0.3\textwidth]{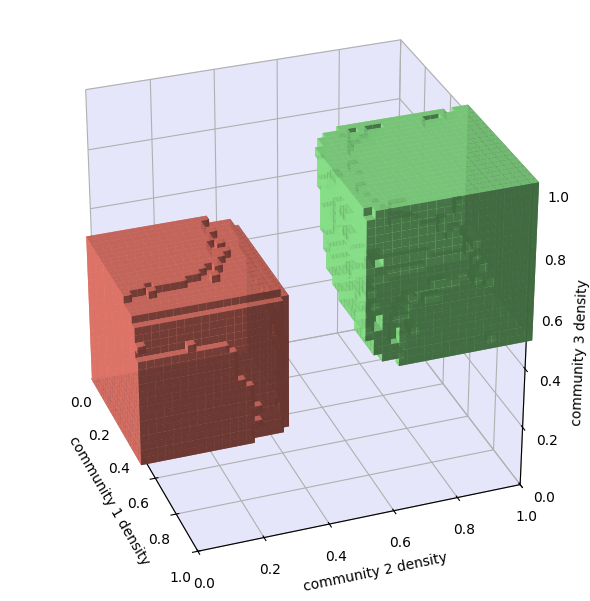}

    \caption{Several slices of 6-dimensional parameter space for 3-block SBMs. The axis coordinates denote the within-group edge densities for each group, and the between-group densities are fixed at 0.1, 0.3, and 0.45 respectively.}
    \label{fig:3dslices}
\end{figure}

\begin{figure}
    \centering
    \includegraphics[width=0.3\linewidth]{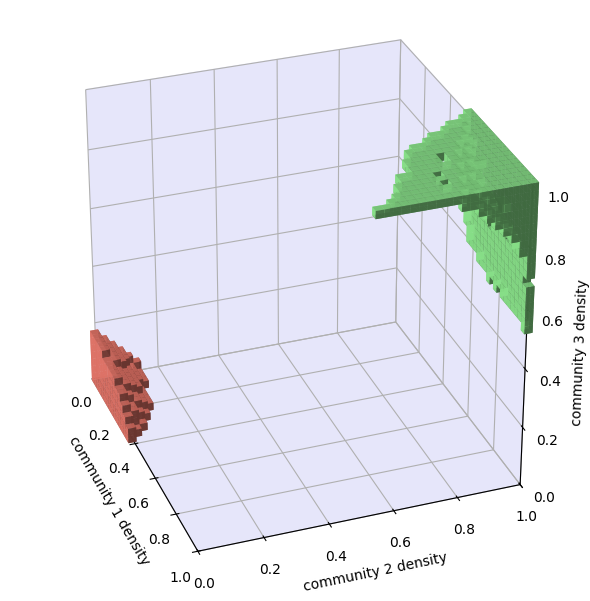}
    \caption{A slice of parameter space for 3-block SBMs where the between-group densities are different. For axis coordinates $p,q,r$, this plot shows regions of information loss for the class of SBMs $P = \begin{bmatrix}
        p & 0.3 & 0.45 \\ 0.3 & q & 0.1 \\ 0.45 & 0.1 & r
    \end{bmatrix}$. Notice that the green (full information) region has boundary faces that closely resemble the corresponding boundary faces in Fig.~\ref{fig:3dslices}, while the red (information loss) region is less sensitive.}
    \label{fig:3dslices_uneven}
\end{figure}

We show a few 3-dimensional ``slices'' of the 6-dimensional parameter space for 3-block SBMs in Fig.~\ref{fig:3dslices}. Notice that they act in an analogous way to the 2-dimensional slices of the $2\times 2$ case, where the dense region ``inverts'' from concave to convex as cross-edges increase.

Note that these slices don't show the whole picture, since they don't capture the regimes where the cross (between-community) blocks have different densities. An example is shown in Fig.~\ref{fig:3dslices_uneven}. In these cases, the (dense) full information regime resembles different (dimension-wise) versions of the same class of shapes, reflecting the PSD condition described in Observation~\ref{obs:PSD_condition}. In contrast, the information loss region is less sensitive to the maximum cross-group density. 

\subsection{Equivalence Classes for $k$-Block SBMs}
\label{appx:numericaleqclasses}

In general, the equivalence classes induced by embeddings for $k$-block SBMs follow the same structure as the 2-block case. For an initial $k$-block SBM $P_0$ with embedding $\hat K$, its equivalence class takes the form
\[ F_{\hat K} = \left\{ P_0 + \Delta \mathrm N\right\}\]
where $\mathrm N$ is a $k\times k$ symmetric matrix (generalizing $\begin{bmatrix}
    \eta & 1 \\ 1 & \frac{1}{\eta}
\end{bmatrix}$ from the 2-block case) and $\Delta$ is a real number, subject to the constraints that $P_0 + \Delta \mathrm N$ is non-negative and 
\begin{equation}
    \Delta \leq \inf_{\delta} (P_0 + \delta \mathrm N \text{ is positive semi-definite}).
    \label{eq:PSD_highorder}
\end{equation}

We observe the following in $\mathrm N$: 

\begin{empobservation}[Low-rank-like densification rate]
    The entries of $\mathrm N$ satisfy 
    \[\mathrm N[i,i] \cdot \mathrm N[j,j] \approx \mathrm N[i,j]^2 \text{ for all pairs of blocks } i \neq j.\]
    \label{obs:lowrank_highorder}
\end{empobservation}

This suggests that the ``slope'' $\mathrm N$ is only has $k-1$ degrees of freedom: namely, the distinct densification rates of each community up to scaling. This low-rank-like behavior\footnote{$\mathrm N$ is not, in fact, low rank. Squaring the diagonal entries would make it low-rank: $\mathrm N' = v^\top v$, where $v$ is the vector of diagonal entries.} is intriguing and we leave it open as a very promising direction for exact characterizations of high-order SBM equivalence classes. 

\subsubsection{Illustrative examples}
\label{appx:numericaleqclasses_examples}
We give two illustrative examples of sets of SBMs that embed to identical representations. 

For three equal-sized blocks, the following class of SBMs all receive the same embedding: 
\[\left\{\begin{bmatrix}
    0.2 & 0 & 0 \\ 0 & 0.3 & 0 \\ 0 & 0 & 0.4
\end{bmatrix} + \Delta \begin{bmatrix}
    1.18 & 1.11 & 1 \\ 1.11 & 1.03 & 0.93 \\ 1 & 0.93 & 0.84
\end{bmatrix} : \Delta \in [0, 0.386] \right\}\]

Note that, as in the 2-block case, the sparser communities receive a larger densifying rate (1.18 for the sparsest community compared to 0.84 for the densest). In this case, $\Delta=0.386$ is the extremal value that satisfies Condition~(\ref{eq:PSD_highorder}) above.

For three different-sized blocks of size $(0.2,0.3,0.5)$, the following class of SBMs all receive the same embedding: 

\[\left\{\begin{bmatrix}
    0.2 & 0 & 0 \\ 0 & 0.2 & 0 \\ 0 & 0 & 0.2
\end{bmatrix} + \Delta \begin{bmatrix}
    1.15 & 1.08 & 1 \\ 1.08 & 1.01 & 0.93 \\ 1 & 0.93 & 0.86
\end{bmatrix} : \Delta \in [0, 0.441] \right\}\]

Once again, as in the 2-block case, the smaller communities have a larger densifying rate (1.15 for the smallest community compared to 0.86 for the largest).

The property described in Observation~\ref{obs:lowrank_highorder} also holds in each of these examples. 





\newpage
\section{Unbalanced Densification in Real-World Graphs}
\label{appx:fb100}

We conduct a very simple experiment to show that the patterns of disproportionate densification are not simply an artifact of stochastic block model graphs: on real graphs, the same patterns occur, sometimes to an even \emph{greater} extent than in SBMs. 

\subsection{Data} 
We use the Facebook100 dataset \cite{traud2012social} (sourced from the Network Data Repository \cite{nr}), which consists of a set of social networks from 100 universities in the United States from September 2005. The data includes information about the class year of each node. The networks are generally strongly assortative by class year. 

For each network, we consider the induced subgraph consisting of nodes with class year of 2005 or 2008, which typically have the highest degree of homophily (due to the separation in class years). We denote 2005  community 1 and 2008 community 2. 

\subsection{Embedding and Link Prediction}
We embed each of the 100 networks using Algorithm~\ref{alg:embedding_generic}. For all networks, we embed into 6 dimensions, and subsample vertices uniformly in batches of $\max(200, n/15)$ (where $n$ is the size of the network). We run with the Adam for 10,000 iterations, with a dynamic learning-rate scheduler and early stopping when the learning rate falls below 1e-6. 

For link prediction, we assign each link $(i,j)$ a predicted probability of $p_{ij} = \sigma(\langle \omega_i, \omega_j\rangle)$, corresponding to the simple algorithm described in Section~\ref{sec:link-prediction}.

\subsection{Relative densification}

For the original network, we compute the matrix $P_\text{initial}$, where $P_\text{initial}[i,j]$ is the average edge density between groups $i$ and $j$ (that is, the number of edges where one end is in group $i$ and the other in $j$, divided by the total possible number of such edges). We compute the \emph{learned} densities by simply summing the probabilities $p_{ij}$ within each block and dividing by the number of such pairs, yielding $P_\text{learned}$.

Then, for each network, we compute its \emph{densification ratio} $d_r$: the ratio of densification in group 2 versus group 1.
\[d_r = \frac{P_\text{learned}[2,2] - P_\text{initial}[2,2]}{P_\text{learned}[1,1] - P_\text{initial}[1,1]} \]

Additionally, we compute its \emph{predicted} densification ratio using only $P_\text{initial}$. Given only the average densities of the original graph, we can use Theorem~\ref{thm:eqclasses} to compute $\eta$ corresponding to the equivalence class containing $P_\text{initial}$. We expect $d_r$ to be close to $\frac{1}{\eta^2}$. 

In fact, the empirical $d_r$ and the theoretically predicted $\frac{1}{\eta^2}$ correlate strongly, as shown in Fig.~\ref{fig:eta-rd-corr}. That is, networks that our theory predicts will have large disparate effects in link prediction do, in practice, have those large effects. This is remarkable given that the prediction is based on \emph{only three statistics of the original network}: the average density of 2005-2005, 2005-2008, and 2008-2008 edges. 

\begin{figure}
    \centering
    \includegraphics[width=0.5\linewidth]{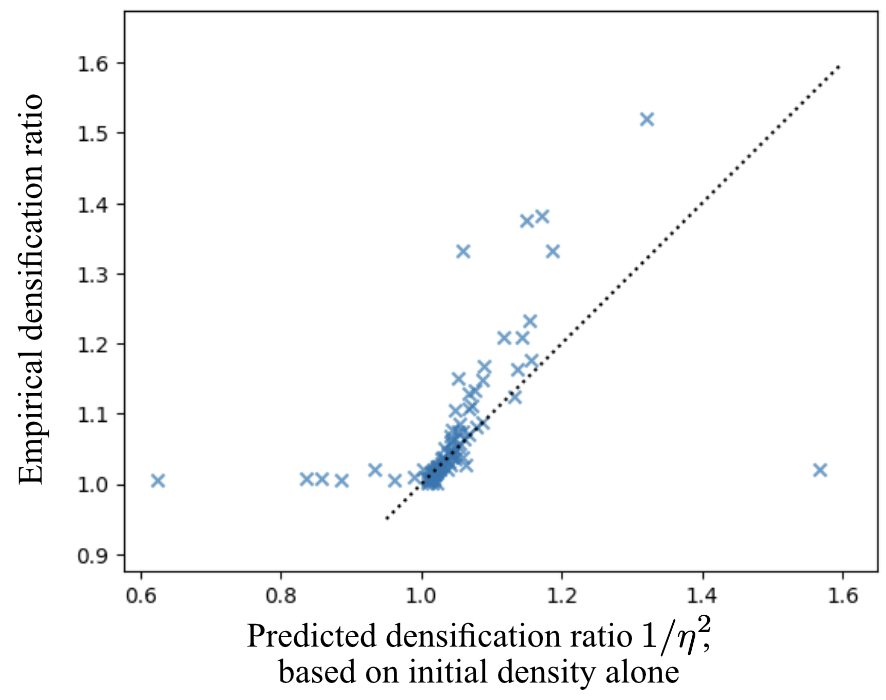}
    \caption{Correlation between predicted and empirical link over-prediction rates in the Facebook100 dataset.}
    \label{fig:eta-rd-corr}
\end{figure}

We also plot the empirical densification ratio against the size and sparsity differences of the initial network, grouped by the average degree of the network, shown in Fig.~\ref{fig:fb100-pattern}. In these real graphs, smaller and sparser communities are also consistenly over-densified relative to larger and denser ones.  

\begin{figure}
    \centering
        \begin{subfigure}[t]{0.49\linewidth}
        \includegraphics[width=0.95\linewidth]{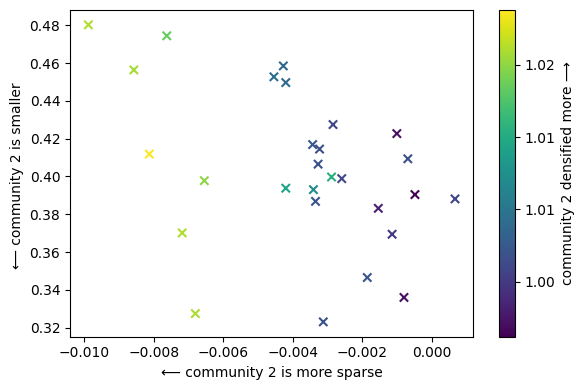}
        \caption{Networks with average edge density below 0.01}
    \end{subfigure}
    \hfill 
    \begin{subfigure}[t]{0.49\linewidth}
        \includegraphics[width=\linewidth]{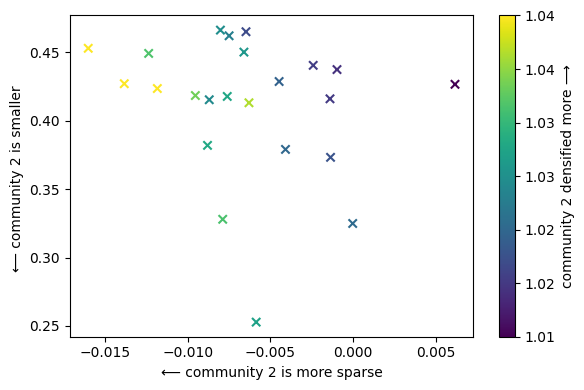}
        \caption{Networks with average edge density in [0.01, 0.02)}
    \end{subfigure}
    \vspace{1em}
    
    \begin{subfigure}[t]{0.49\linewidth}
        \vspace{0pt}
        \includegraphics[width=\linewidth]{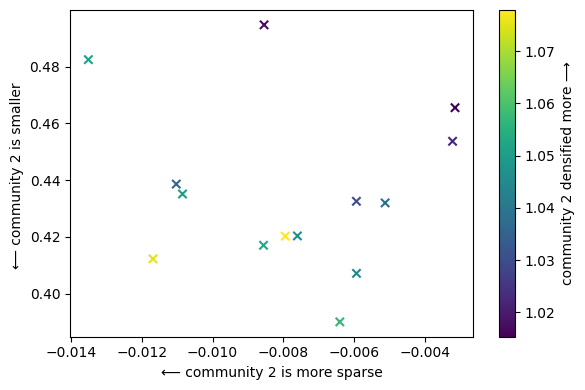}
        \caption{Networks with average edge density in $[0.02, 0.03)$}
    \end{subfigure}
    \caption{When embedding the Facebook100 networks, smaller and sparser communities are consistently over-densified. Here, we plot the densification ratio $d_r$ of each network (color) against the ratio of community sizes and difference in sparsity, stratified by the overall density of the graph.}
    \label{fig:fb100-pattern}
\end{figure}

\newpage 

\section{Community detection is robust except in $\Pi_s$}
\label{appx:tasks}
One common task is to do some form of classification or clustering of the nodes into communities. An easy consequence of Theorem~\ref{thm:embedding-limit-regimes} is that, in parameter regimes $\Pi_m$ and $\Pi_d$, the community information (that is, which nodes are in which communities) is fully encoded in the loss-minimizing embedding: this information can be read off from the block-constant structure of $K$. Only in region $\Pi_s$ is the community structure lost. 

This is qualitatively similar to the Kesten-Stigum information-theoretic threshold \cite{decelle2011asymptotic}, which describes conditions where SBM communities cannot be recovered from a \emph{single instance} of the graph at rates better than chance. The KS threshold similarly implies that heterophilic communities cannot be detected in certain cases, but detection scales with the size of the graph $n$ (unlike our result, which is explicitly based on the generating model of the graph rather than a single instance). 

Empirically, it is known that standard graph representation methods tend to do poorly on heterophilic graphs \cite{gao2018community}, and explicit algorithms for heterophilic graph learning tend to rely on somewhat specialized techniques; see \cite{luan2024heterophilic}. The information loss that we show in $\Pi_s$ gives additional theoretical justification for the poor performance of inner-product-based methods. 





\end{document}